\documentclass[1p]{article}

\usepackage{graphicx}
\usepackage{amssymb,amsmath}
\usepackage{pstricks,psfrag}
\usepackage{url}

\long\def\invis#1{}

\newtheorem{cor}{Corollary}
\newtheorem{theorem}{Theorem}
\newtheorem{lem}{Lemma}
\newtheorem{prop}{Proposition}
\newtheorem{proper}{Property}

\newcommand{\corref}[1]{Corollary~\ref{cor:#1}}
\newcommand{\figref}[1]{Figure~\ref{fig:#1}}

\newcommand{\propref}[1]{Proposition~\ref{prop:#1}}
\newcommand{\properref}[1]{Property~\ref{proper:#1}}
\newcommand{\secref}[1]{Section~\ref{sec:#1}}
\newcommand{\thmref}[1]{Theorem~\ref{thm:#1}}
\renewcommand{\eqref}[1]{(\ref{eq:#1})}

\newcommand{\all}{\mathrm{all}}
\newcommand{\row}{\mathrm{row}}
\newcommand{\col}{\mathrm{col}}

\newcommand{\probexist}{{\sc Existence}}
\newcommand{\probmax}{{\sc MaxInfeasible}}
\newcommand{\probmin}{{\sc MinUniversal}}
\newcommand{\MINISAT}{\textsc{MiniSat}}
\newcommand{\SUGAR}{\textsc{Sugar}}

\newenvironment{proof}{\medskip
\noindent{\bf Proof:}}{\quad $\Box$\medskip}

\title{On A Generalization of ``Eight Blocks to Madness''}
\author{Kazuya Haraguchi}
\date{Faculty of Commerce, Otaru University of Commerce, Japan\\ \tt{haraguchi@res.otaru-uc.ac.jp}
}

\begin{document}
\maketitle

\begin{abstract}
    We consider a puzzle such that a set of colored cubes is given as an instance. 
    Each cube has 
    unit length on each edge 
    and its surface is colored so that 
    what we call the Surface Color Condition is satisfied.
    Given a palette of six colors, the condition requires that each face should have 
    exactly one color and all faces should have different colors from each other. 
    The puzzle asks to compose a $2\times 2\times 2$ cube
    that satisfies the Surface Color Condition 
    from eight suitable cubes in the instance. 
    Note that cubes and solutions have 30 varieties respectively. 
    In this paper, we give answers to three problems on the puzzle:
    {\bf (i)} For every subset of the 30 solutions,
    is there an instance that has the subset exactly
    as its solution set?
    {\bf (ii)} Create a maximum sized infeasible instance (i.e., one having no solution). 
    {\bf (iii)} Create a minimum sized universal instance (i.e., one having all 30 solutions).
    We solve the problems with the help of a computer search. 
    We show that the answer to (i) is no.  
    For (ii) and (iii), 
    we show examples of the required instances,
    where their sizes are 23 and 12, respectively. 
    The answer to (ii) 
    solves one of the open problems
    that were raised in [E. Berkove et al., ``An Analysis of the (Colored Cubes)$^3$ Puzzle,'' {\em Discrete Mathematics\/}, {\bf 308\/} (2008) pp.~1033--1045]. 
\end{abstract}

\invis{
}


\section{Introduction}
\label{sec:intro}

Various kinds of mathematical puzzles have been invented, 
which provide us with not only recreation but also interesting research problems. 
Many types of puzzles such as pencil puzzles~\cite{andersson2009hashiwokakero,DemaineOUU13,HO.2013,kolker2012kurodoko,YS.2003} 
and even video games~\cite{ADGV.2014,DHL.2003,W.2014} 
have been shown to be NP-hard,
which seems to be a common property that
attractive puzzles should have. 
Hearn and Demaine collected
major complexity results in \cite{hearn2009games}. 

Puzzle instance creation 
is one of the possible future research directions 
in mathematical puzzles. 
Ordinary people do not have interest in computational complexity.
They just like to play more and more addictive and challenging puzzle instances. 
Hence it is meaningful to exploit how to create puzzle instances that have desired properties. 
Nevertheless, 
fairly little literature deals with puzzle instance creation
(e.g., {\sc Futoshiki}~\cite{H.2013} and {\sc BlockSum}~\cite{HAM.2012}). 
%
This is due to its difficulty.
Creating a puzzle instance is harder than just solving it in general
because creation process involves solving. 
To create an instance, we need not only to check whether a candidate instance is solvable
but also to examine whether it has desired properties (e.g., solution uniqueness)
or even to search the instance space for better candidates.

With these in mind, 
we explore how to create an instance
for simple puzzles. 
In this paper, we take up a certain generalization of
a classical puzzle called ``Eight Blocks to Madness.''
We present a collection of our research results on this puzzle.  
Specifically, we give answers to three problems named
\probexist, \probmax, and \probmin\  
that are described below.

In the puzzle that we consider, we deal with colored cubes. 
Suppose that a palette of six colors is given.
A puzzle solver is given a set of $m$ colored cubes as an {\em instance\/}.
Each 
cube 
has unit length on each edge 
and the surface is colored so that the following
{\em Surface Color Condition\/} is satisfied.
\begin{description}
\item[Surface Color Condition:]
  {\em Each face 
    has 
    exactly one color
    and all faces
    have 
    different colors from face to face.\/}
\end{description}
Hence, 30 {\em varieties\/} are possible for colored cubes.
The instance may contain multiple cubes for some varieties
or may contain no cubes for other varieties. 
The puzzle asks to compose a $2\times2\times 2$ cube 
by using eight suitable cubes in the instance
so that the $2\times2\times2$ cube satisfies the Surface Color Condition.  
We call a $2\times 2\times 2$ cube whose surface is colored in that way a {\em solid\/}. 
Of course there are 30 varieties of solids.
Once an instance is given, 
the solids that can be constructed 
are uniquely determined. 
Note that all the 30 varieties do not necessarily appear.

To create a puzzle instance, we fix the solution set at first 
and then search for an instance that has this solution set,
as has been done in creation of pencil puzzle instances~\cite{H.2013,HAM.2012}. 
Let us denote the set of all the 30 varieties by $V_\all$
and an arbitrary subset of $V_\all$ by $V$. 
An instance is {\em $V$-generatable\/} 
if a solid of any variety in $V$ can be composed and 
no solid of any variety out of $V$ can be composed anyhow. 
In our creation process, 
we first decide a subset $V$ of $V_\all$ 
and then construct a $V$-generatable instance. 
\invis{
Specifically, we first decide a subset $V\subseteq V_\all$ of varieties
and then construct a  {\em $V$-generatable\/} instance,
that is, a set so that 
a solid of any variety $v\in V$ can be composed and 
no solid of any variety $v'\in V_\all\setminus V$ can be composed anyhow. 
}
We come up with the following problem. 
\begin{description}
  \item[\underline{Problem \probexist}]
  Is there a $V$-generatable instance 
  for all possible
  $V\subseteq V_\all$?
\end{description}
We show that, unfortunately, the answer is no. 
This suggests that a puzzle creator should choose 
the solution set $V$ appropriately
so that a $V$-generatable instance exists. 

Next, we observe two extreme cases:
$V=\emptyset$ and $V=V_\all$.
Let us call an $\emptyset$-generatable instance {\em infeasible\/} and
a $V_\all$-generatable instance {\em universal\/}. 
It is easy to construct examples of these types of instances. 
For example, 
an instance containing seven or fewer cubes is infeasible. 
An instance that contains eight or greater cubes for all the 30 varieties is universal. 
It is interesting to explore a larger infeasible instance
and a smaller universal instance
since, intuitively, 
the more cubes and the more varieties an instance contains,
the more varieties of solids are expected to be composable.  
The following problems arise. 
\begin{description}
\item[\underline{Problem \probmax}]
  Create an infeasible instance of the maximum size. 
\end{description}
\begin{description}
\item[\underline{Problem \probmin}]
  Create a universal instance of the minimum size. 
\end{description}
We give answers to these problems by
showing examples of a maximum infeasible instance
and a minimum universal instance,
where their sizes are 23 and 12, respectively. 
We emphasize that \probmax\ solves 
one of the open problems in \cite{Berkove.2008}.

To tackle these problems,
we utilize the power of a computer search. 
Our approach can be regarded as ``computer-assisted proof,''
which has been used to solve various significant problems in
discrete mathematics;
e.g., Four color theorem~\cite{AH.1977-1,AH.1977-2,AHK.1977}, 
Sudoku critical set~\cite{MTC.2013}. 
We formulate the problems
by {\em Constraint Programming\/} ({\em CP\/}) or
by {\em Integer Programming\/} ({\em IP\/}). 
Once the models are established,
we may be able to solve the problems
by using state-of-the-art computational softwares.

The paper is organized as follows. 
We review previous work related to the puzzle in \secref{rel}
and prepare terminologies and notations in \secref{prel}. 
We present the answers to the three problems in \secref{answer},
along with their CP or IP formulations.
We conclude the paper in \secref{conc},
describing open problems. 

\section{Related Work}
\label{sec:rel}

The original version of ``Eight Blocks to Madness''
was invented by Irishman Eric Cross in 1960's 
and was issued by Austin Enterprises of Ohio~\cite{ROB}.  
The puzzle instance has exactly eight cubes
that are of different varieties from each other. 
Thus a solver has only to arrange the eight given cubes into a solid. 
Kahan~\cite{Kahan.1972} showed that the solid variety composable from this original instance is unique. 
Sobczyk~\cite{Sobczyk.1974} studied a slight generalization of
the original puzzle so that an instance
can be a set of {\em any\/} eight cubes,
repetition of the same variety cubes being allowed. 
He showed that at most six varieties of solids are composable from such an instance. 
Our puzzle is a generalization of these previous puzzles
in the sense that an instance is a set of any $m$ cubes. 

Berkove et al.~\cite{Berkove.2008} already studied a further generalization. 
In their puzzle, which they call ``(Colored Cubes)$^3$ Puzzle,''
a puzzle solver is given a set of $m$ cubes as an instance
and is asked to compose an $n\times n\times n$ solid,
where $m\ge n^3$ is assumed. 
They prove that a solid is always composable when $n\ge3$. 
The scenario of the proof is described as follows: 
they show that, when $n\ge3$, an $n\times n\times n$ solid is composable
iff the instance contains such a subset of eight cubes
that can be arranged into a $2\times2\times2$ solid. 
They show that every instance of no less than $3^3=27$ cubes
has such a subset. 
Such a subset is significant because, to compose an $n\times n\times n$ solid,
it is most ``difficult'' to find eight cubes from the instance 
that are assigned to the corners of the solid;
these cubes form a $2\times2\times2$ solid.
We use the term ``difficult'' to mean that
cubes that can be assigned to corners
are most restricted since the three faces are exposed to the outside. 
On the other hand, it is comparatively easy to find 
cubes to be assigned to other parts of the solid.
For example, we can assign any cube to the inside
since no face is exposed to the outside. 
Thus, the $n=2$ case plays a significant role in the analysis of the $n\ge3$ case. 
This is why we concentrate on the $n=2$ case. 

Colored cube puzzles appear to have been
introduced to recreational mathematics
by P.A.~MacMahon in 1920's~\cite{MacMahon.1921,MacMahon.WWW}. 
He invented the first puzzle that deals with colored cubes. 
In his puzzle, a puzzle solver is given a set of 29 cubes
of different varieties except a certain variety, say $v$. 
The solver is asked to compose a $2\times2\times2$ solid
of the variety $v$ so that the {\em Domino Condition\/} is satisfied,
i.e., in the inside of the solid,
two faces touching each other should have the same color. 
The research history of this puzzle is summarized 
in M.~Gardner's famous ``Fractal Music, Hypercards and More...''~\cite{Gardner.1991}. 
The most outstanding way to solve the puzzle
was developed by J.H.~Conway.  
He arranged the 30 varieties of cubes
in the non-diagonal entries of a $6\times 6$ grid, as in \figref{conway}. 
When $v$ is the variety in the row $i$ and in the column $j$,
we can compose the desired solid from eight cubes,
one each of the varieties in row $j$
and column $i$, not including the $(j,i)$-variety. 

\begin{figure}[!t]
  \begin{center}
    \begin{tabular}{cccccc}
      \includegraphics[width=1.6cm,keepaspectratio]{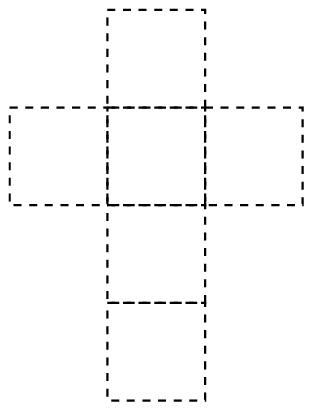}&
      \includegraphics[width=1.6cm,keepaspectratio]{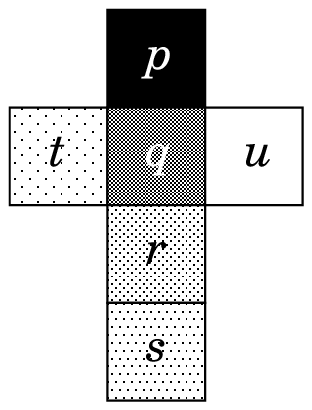}&
      \includegraphics[width=1.6cm,keepaspectratio]{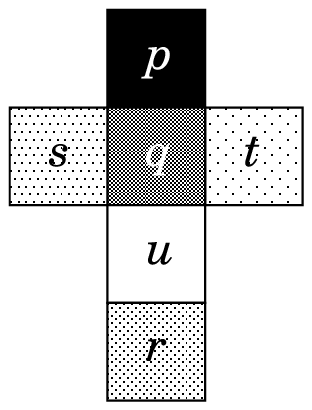}&
      \includegraphics[width=1.6cm,keepaspectratio]{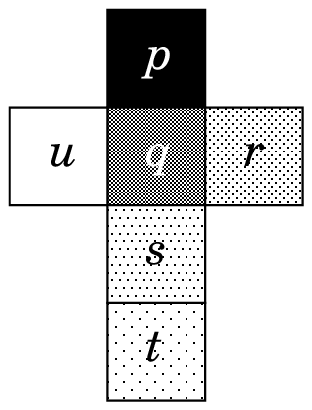}&
      \includegraphics[width=1.6cm,keepaspectratio]{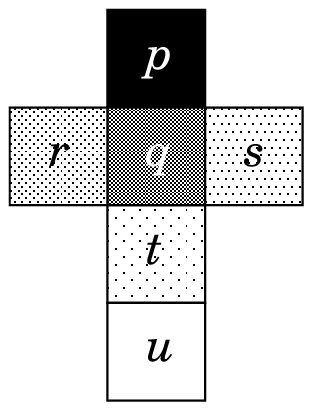}&
      \includegraphics[width=1.6cm,keepaspectratio]{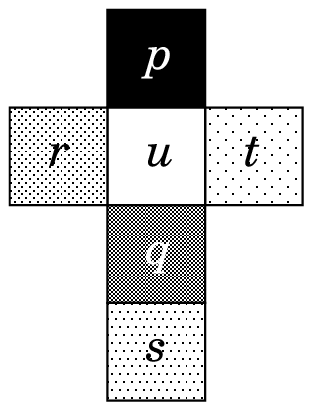}\\
      \includegraphics[width=1.6cm,keepaspectratio]{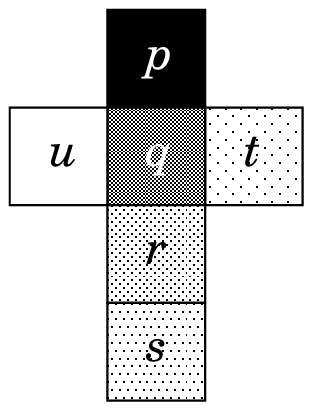}&
      \includegraphics[width=1.6cm,keepaspectratio]{Cube_blank.eps}&
      \includegraphics[width=1.6cm,keepaspectratio]{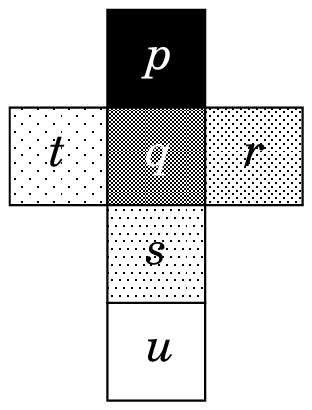}&
      \includegraphics[width=1.6cm,keepaspectratio]{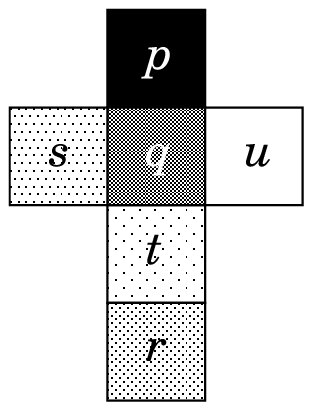}&
      \includegraphics[width=1.6cm,keepaspectratio]{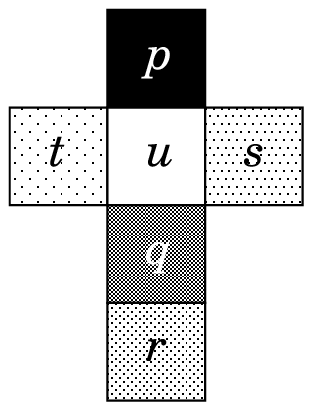}&
      \includegraphics[width=1.6cm,keepaspectratio]{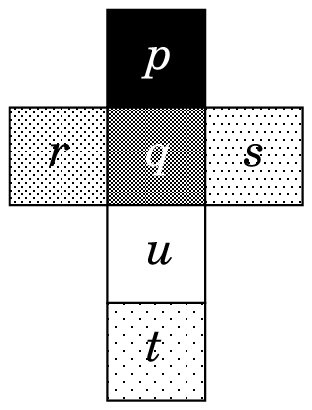}\\
      \includegraphics[width=1.6cm,keepaspectratio]{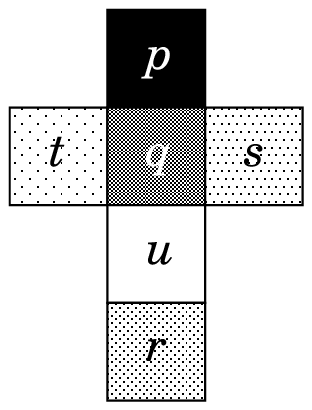}&
      \includegraphics[width=1.6cm,keepaspectratio]{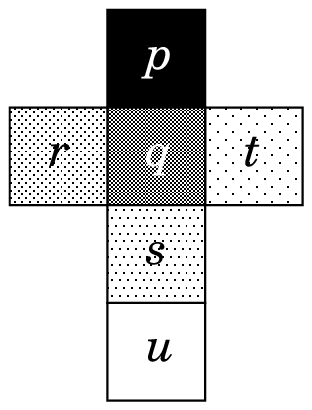}&
      \includegraphics[width=1.6cm,keepaspectratio]{Cube_blank.eps}&
      \includegraphics[width=1.6cm,keepaspectratio]{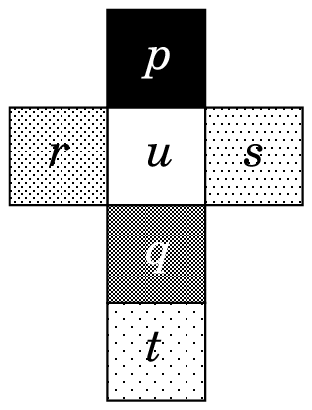}&
      \includegraphics[width=1.6cm,keepaspectratio]{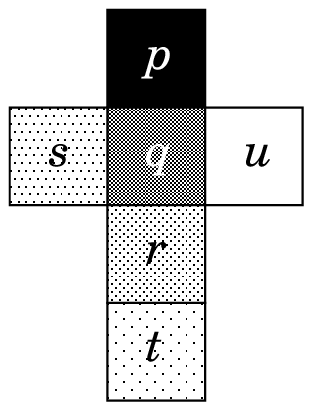}&
      \includegraphics[width=1.6cm,keepaspectratio]{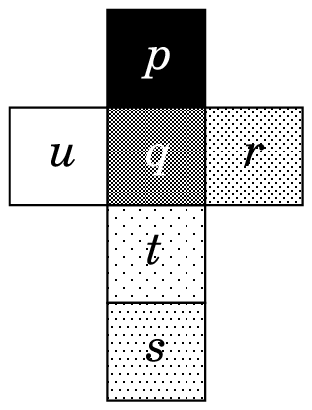}\\
      \includegraphics[width=1.6cm,keepaspectratio]{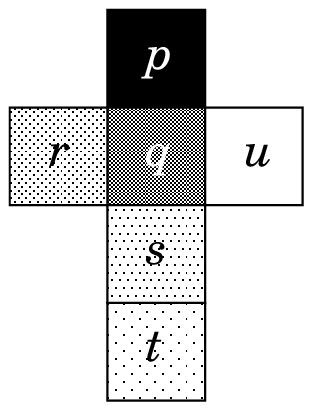}&
      \includegraphics[width=1.6cm,keepaspectratio]{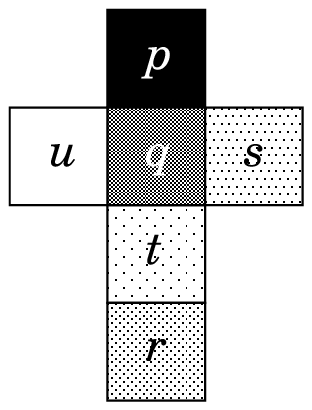}&
      \includegraphics[width=1.6cm,keepaspectratio]{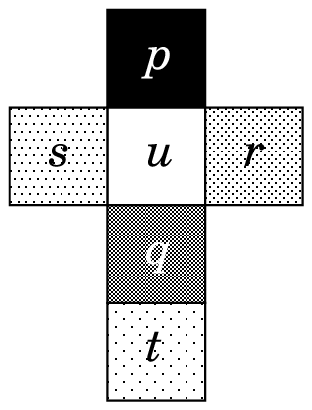}&
      \includegraphics[width=1.6cm,keepaspectratio]{Cube_blank.eps}&
      \includegraphics[width=1.6cm,keepaspectratio]{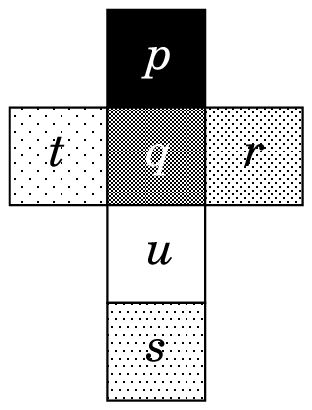}&
      \includegraphics[width=1.6cm,keepaspectratio]{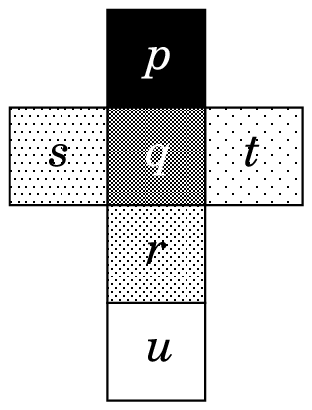}\\
      \includegraphics[width=1.6cm,keepaspectratio]{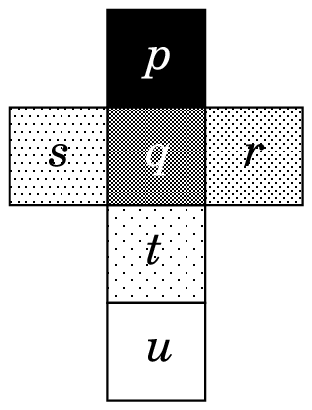}&
      \includegraphics[width=1.6cm,keepaspectratio]{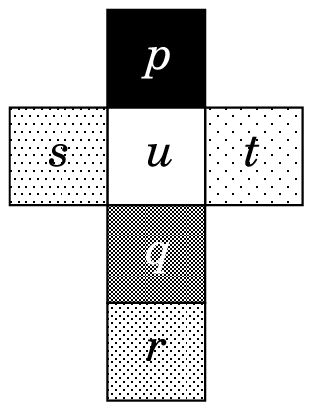}&
      \includegraphics[width=1.6cm,keepaspectratio]{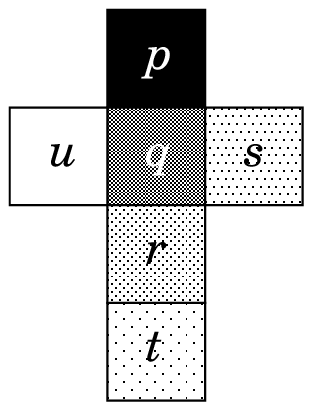}&
      \includegraphics[width=1.6cm,keepaspectratio]{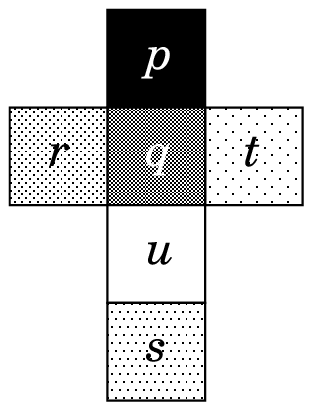}&
      \includegraphics[width=1.6cm,keepaspectratio]{Cube_blank.eps}&
      \includegraphics[width=1.6cm,keepaspectratio]{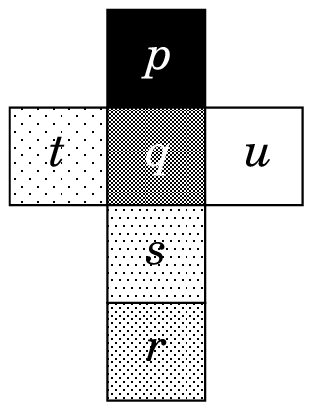}\\
      \includegraphics[width=1.6cm,keepaspectratio]{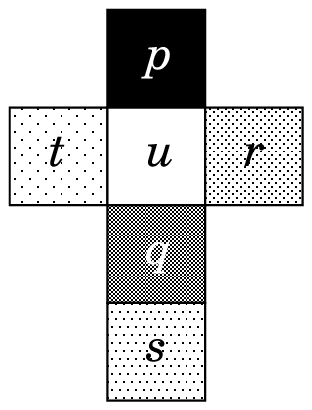}&
      \includegraphics[width=1.6cm,keepaspectratio]{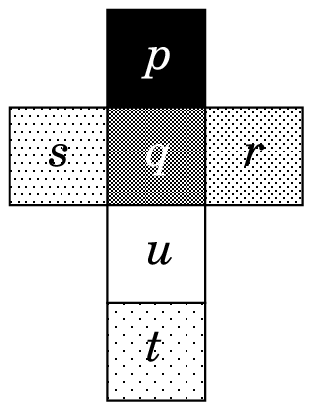}&
      \includegraphics[width=1.6cm,keepaspectratio]{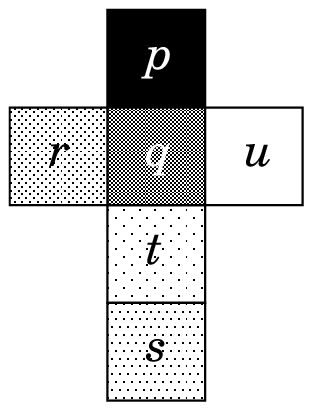}&
      \includegraphics[width=1.6cm,keepaspectratio]{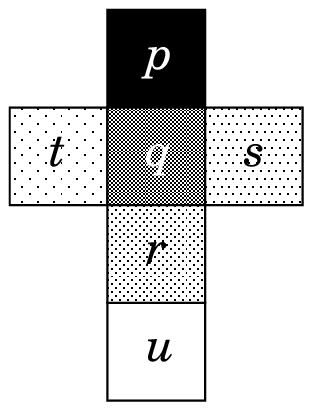}&
      \includegraphics[width=1.6cm,keepaspectratio]{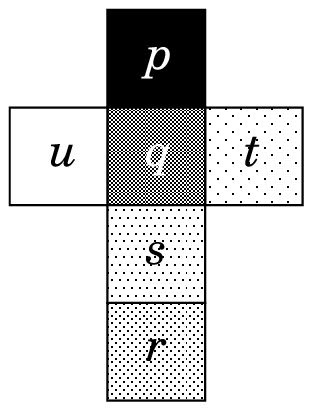}&
      \includegraphics[width=1.6cm,keepaspectratio]{Cube_blank.eps}
    \end{tabular}
    \caption{Conway's table of the 30 cube varieties. Each variety is represented by how the 6 colors are assigned to the development of a cube. 
      The six colors are denoted by $p,q,r,s,t,u$.}
    \label{fig:conway}
  \end{center}
\end{figure}

Our puzzle is different from  MacMahon's in three points:
in our puzzle, an instance is a multi-set of $m$ cubes,
the solid variety to be constructed is not specified to the solver,
and the Domino Condition is not imposed on the inside of the solid. 
Nevertheless,
we will use Conway's table in the analysis of our puzzle 
since it has interesting properties
that let us understand the analysis 
more easily.

Among other types of colored cube puzzles,
{\em Instant Insanity\/}\footnote{``Instant Insanity'' was
  originally trademarked by Parker Brothers in 1967.
  The trademark is now owned by Winning Moves, Inc.}
must be the most well-known puzzle. 
Recently Demaine et al.~analyzed
the computational complexity of its variations in \cite{DDEMU.2013},
where the research history is
well-covered.

\section{Preliminaries}
\label{sec:prel}

\subsection{Our Puzzle}
A cube has six faces, 12 edges and eight corners. 
The variety of a cube is specified by
how the six colors are assigned to its faces. 
We index the 30 varieties by means of Conway's table in \figref{conway}.
We denote the variety in the row $i$ and in the column $j$
by $(i,j)$. 
For a natural number $n$, let $[n]=\{1,2,\dots,n\}$. 
We denote the set of all the 30 varieties by $V_\all$,
i.e., $V_\all=\{(i,j)\in[6]\times[6]\mid i\ne j\}$. 
We call a cube of variety $(i,j)$ an {\em $(i,j)$-cube\/}. 

In order to explain the properties of Conway's table,
we introduce the notion of corner triple. 
For each corner of cube,
we define 
a
{\em corner triple\/} as an ordered triple
of three colors around the corner such that
the head element is set to
the color of the lexicographically smallest
letter
among them
and the colors are taken
in the clockwise order around the corner.  
A cube (and thus a variety)
has eight corner triples.
We denote the set of eight corner triples of variety $(i,j)$
by $T_{i,j}$.
We show examples as follows. 
\begin{align*}
T_{1,2} =& \{(p,q,t),(p,s,u),(p,t,s),(p,u,q),(q,r,t),(q,u,r),(r,s,t),(r,u,s)\},\\
T_{2,1} =& \{(p,t,q),(p,u,s),(p,s,t),(p,q,u),(q,t,r),(q,r,u),(r,t,s),(r,s,u)\}. 
\end{align*}
Observe that $40$ corner triples are possible in all. 
\invis{
This is because, although
there are $6\times5\times4=120$ ordered triples of three colors, 
we deal with triples like $(r,p,q)$ and $(q,r,p)$ as $(p,q,r)$
since $p$ is lexicographically smaller than $q$ and $r$. 
}
Two corner triples are the 
{\em mirror images\/}
of each other if both contain the same three colors but are different ordered triples.
For example, $(p,q,t)$ and $(p,t,q)$ are the mirror images of each other. 
Two varieties $(i,j)$ and $(i',j')$ are
the {\em mirror varieties\/} of each other 
if, for each triple in $T_{i,j}$, 
its mirror image belongs to $T_{i',j'}$. 
For example, $(1,2)$ and $(2,1)$ are the mirror varieties of each other. 

We say that two different varieties are {\em compatible\/}
(resp., {\em incompatible\/}) with each other
if they share at least one corner triple
(resp., no corner triple). 
For convenience, we may say that two cubes are compatible (resp., incompatible)
when their varieties are compatible (resp., incompatible) with each other. 
For a variety $(i,j)$,
let $V_{i,j}$ (resp., $\bar{V}_{i,j}$) denote the set of varieties compatible
(resp., incompatible) with $(i,j)$. 
Note that neither $V_{i,j}$ nor $\bar{V}_{i,j}$ contains the variety $(i,j)$ itself. 
Every variety in $V_{i,j}$ can be obtained by changing the color assignment 
of $(i,j)$ by either the ``swap'' operation 
or the ``rotation'' operation~\cite{Berkove.2008}. 
By the swap operation, we mean
the exchange of two colors
of adjacent faces
that share one edge with each other.
By the rotation operation, we mean
the rotation of colors
of the three faces
around a corner so that the corresponding corner triple is unchanged. 
The former yields 12 varieties
distinct from $(i,j)$ 
and the latter yields eight varieties
distinct from $(i,j)$. 
The variety set $\bar{V}_{i,j}$ consists of
the mirror variety of $(i,j)$
and the mirrors of the eight varieties that are generated by performing
the rotation operation on $(i,j)$. 

\begin{prop}
 \label{prop:share}
  {\bf (Proof of Lemma 2.7 in \cite{Berkove.2008})}
  We have $|V_{i,j}|=20$ and $|\bar{V}_{i,j}|=9$ 
  for any variety $(i,j)$. 
\end{prop}
Moreover, whenever two varieties are compatible,
they share exactly two corner triples. 
\begin{prop}
 \label{prop:corner}
  {\bf (Lemma 2.7 in \cite{Berkove.2008})}
  Two different varieties share either zero or two corner triples with each other.
\end{prop}

Now we present the properties of Conway's table as follows. 
%
\begin{proper}
\label{proper:mirror}
  The mirror variety of $(i,j)$ is $(j,i)$.
\end{proper}
\begin{proper}
\label{proper:noshare}
  The five varieties on the same row or on the same column
  are incompatible with each other. 
\end{proper}

From $|\bar{V}_{i,j}|=9$, we have
\begin{align*}
  \bar{V}_{i,j}=&\{(j,i)\}\cup\{(i,j')\mid j'\in[6]\setminus\{i,j\}\}
  \cup\{(i',j)\mid i'\in[6]\setminus\{i,j\}\}.
\end{align*}
This implies that, to compose a $2\times2\times2$ cube of variety $(i,j)$, 
we cannot use cubes of the mirror variety $(j,i)$
or cubes of the varieties on the row $i$
or on the column $j$,
except $(i,j)$ itself.  

An {\em instance\/} 
is a multi-set of colored cubes. We denote an instance by $I$.
We represent the distribution of varieties in $I$ by a $6\times 6$ matrix,
which we denote by its calligraphic style ${\mathcal I}$. 
Each value ${\mathcal I}_{i,j}$ in the matrix represents
the number of $(i,j)$-cubes in the instance. 
Note that we let ${\mathcal I}_{i,i}=0$ for any $i\in[6]$. 
We call $|I|$ the {\em size of $I$\/}. 
Clearly we have
\[
|I|=\sum_{(i,j)\in V_\all} {\mathcal I}_{i,j}. 
\]

A {\em solid\/} 
is a $2\times2\times2$ cube 
that is composed of eight $1\times1\times1$ cubes
and that satisfies 
the 
Surface Color Condition.
We call a solid {\em an $(i,j)$-solid\/}
when its variety is $(i,j)$. 
We say that an $(i,j)$-solid is {\em composable from $I$\/}
if $I$ contains a subset of eight cubes
that can be arranged into an $(i,j)$-solid.
An instance $I$ is {\em $V$-generatable\/}
if an $(i,j)$-solid is composable from $I$ when and only when $(i,j)\in V$. 
An instance is {\em universal\/} if it is $V_\all$-generatable. 
An instance is {\em infeasible\/} if it is $\emptyset$-generatable.

\subsection{Composability Conditions}

We introduce two necessary and sufficient conditions
for an $(i,j)$-solid to be composable
from a given instance $I$,
using graph terminology. 
We also derive corollaries from these conditions,
which we will use to solve the three problems
introduced in \secref{intro}.  

The first condition is rather straightforward. 
We consider a bipartite graph $B_{i,j}$
such that nodes on one side are cubes in $I$
and nodes on the other side are corner triples in $T_{i,j}$. 
Since $I$ is a multi-set, it may contain multiple cubes for a single variety,
and all of them appear as distinct nodes. 
A pair $(c,\tau)$ in $I\times T_{i,j}$
is joined by an edge whenever $c$ has corner triple $\tau$.  
The existence of edge $(c,\tau)$ indicates that
$c$ can be allocated to the corner having $\tau$. 
\invis{
Let us define a bipartite graph
$B_{i^\ast,j^\ast}=(I\cup T_{i^\ast,j^\ast},F_{i^\ast,j^\ast})$ such that
the nodes are cubes in $I$ and corner triples in $T_{i^\ast,j^\ast}$
and the edge set 
$F_{i^\ast,j^\ast}\subseteq I\times T_{i^\ast,j^\ast}$
is defined as follows:
\[
F_{i^\ast,j^\ast} =\{(c,\tau)\in I\times T_{i^\ast,j^\ast}\mid \textrm{cube\ }c\textrm{\ has\ corner\ triple\ }\tau\}. 
\]
Hence $(c,\tau)\in F_{i^\ast,j^\ast}$ indicates that
cube $c$
can be allocated to the corner having $\tau$. 
Let us denote an $(i,j)$-cube by $c_{i,j}$. 
}
Let the variety of $c$ be $v$. 
From \propref{corner},
the degree of $c$ is zero if $v\in \bar{V}_{i,j}$,
it is two if $v\in V_{i,j}$, and
it is eight if $v=(i,j)$. 
A {\em matching\/} in a graph is a subset of edges such that
no two edges in the subset have an endpoint in common. 
A matching in $B_{i,j}$
represents admissible allocation of cubes to corners,
where each cube in $I$ is allocated to at most one corner, and
each corner is assigned at most one cube.  
Since $|T_{i,j}|=8$,
we have the following condition. 
\begin{prop}
  \label{prop:Bhave}
  An $(i,j)$-solid is composable from an instance $I$
  iff there is an $8$-size matching in the bipartite graph $B_{i,j}$. 
\end{prop}
Clearly, if an 8-size matching exists, then it is a maximum cardinality matching. 
For any subset $T$ of $T_{i,j}$,
let $I(T)$ denote the subset of cube nodes 
that are adjacent to at least one node in $T$.  
That is, $I(T)$ is the subset of cubes in $I$
that have a corner triple in $T$. 
From the Marriage Theorem~\cite{Lovasz.09},
the following corollary is immediate. 
\begin{cor}
  \label{cor:Bhave}
  An $(i,j)$-solid is composable from $I$
  iff $|T|\le|I(T)|$ holds for all possible subsets $T$ of $T_{i,j}$.
\end{cor}

The second condition 
is based on a certain multi-graph 
that is obtained by changing the structure of $B_{i,j}$ as follows;
remove all cube nodes of degrees 0 and 8 and the connecting edges,
and then shrink every cube node of degree 2 and its connecting edges
into a new edge.
As a result, only $T_{i,j}$ remains as the node set.
Each edge is associated with a cube compatible with $(i,j)$ 
and multi-edges may appear.  
We denote this multi-graph by $G_{i,j}=(T_{i,j},E_{i,j})$.
Since all compatible cubes in $I$ appear as edges,
we have $|E_{i,j}|=\sum_{(i',j')\in V_{i,j}} {\mathcal I}_{i',j'}$.

\invis{
It connects two nodes $\tau$ and $\tau'$
whenever
$T_{i,j}\cap T_{i^\ast,j^\ast}=\{\tau,\tau'\}$.
Each edge 
is associated with an $(i,j)$-cube in the instance 
that is compatible with $(i^\ast,j^\ast)$,
that is, $(i,j)\in V_{i^\ast,j^\ast}$. 
Recall \propref{corner} which states
that any $(i,j)$-cube shares exactly two corner triples with $(i^\ast,j^\ast)$. 
}

\begin{prop}
  \label{prop:Ghave}
  An $(i,j)$-solid is composable from an instance $I$
  iff ${\mathcal I}_{i,j}$ is no less than
  the number of tree
  components
  in the graph $G_{i,j}$. 
\end{prop}
\begin{proof}
For the necessity, 
let us denote the number of tree components by $k$. 
Suppose ${\mathcal I}_{i,j}<k$. 
We denote the set of all nodes in the $k$ tree components by $T$
and the set of all cubes appearing in the components as edges by $I'$. 
We have $|I'|=|T|-k$. 
The set $T$ of corner triples 
is a node subset
in the bipartite graph $B_{i,j}$. 
We have $I(T)=I'\cup\{(i,j)\textrm{-cubes\ in\ }I\}$,
where the second cube set in the right side may be a multi-set.
It follows that:
\[
|I(T)|=|I'|+{\mathcal I}_{i,j}=|T|-k+{\mathcal I}_{i,j}<|T|.
\]
By \corref{Bhave},
we cannot compose an $(i,j)$-solid in any way. 

For the sufficiency,
if an $(i,j)$-solid is not composable,
then there is a subset $T$ of $T_{i,j}$ 
such that $|T|>|I(T)|={\mathcal I}_{i,j}+|I'|$,
where $I'$ is the subset of compatible cubes in $I$ that
are adjacent to a node in $T$.
Consider the subgraph of $G_{i,j}$ induced by $T$.
The number of edges in the subgraph is equal to $|I'|$. 
Denoting the number of tree components in the subgraph by $k$,
we have $|I'|\ge|T|-k$. 
Since $|T|>{\mathcal I}_{i,j}+|T|-k$, we have $k>{\mathcal I}_{i,j}$. 
\end{proof}

We give an example of 
how \propref{Ghave} works. 
In \figref{graph-example},
we show graphs $G_{1,2}$ and $G_{2,3}$ for the instance $I$ with nine cubes 
that has the following matrix representation:
\begin{align}
{\mathcal I}=
\left(
\begin{array}{llllll}
\textrm{0} & 2 & 0 & 0 & 0 & 0\\
0 & \textrm{0} & 0 & 0 & 0 & 1 \\
0 & 0 & \textrm{0} & 0 & 1 & 1 \\
0 & 0 & 0 & \textrm{0} & 0 & 0 \\
0 & 0 & 0 & 0 & \textrm{0} & 2 \\
0 & 0 & 0 & 1 & 1 &\textrm{0}
\end{array}
\right).
\label{eq:matrix-1}
\end{align}
In the figure, a node (resp., an edge) is
labeled with its corresponding corner triple
(resp., cube variety). 
We see that a $(1,2)$-solid is composable
since there is only one tree component in $G_{1,2}$ (i.e., an isolated point $(p,s,u)$)
and ${\mathcal I}_{1,2}=2$.
On the other hand, a $(2,3)$-solid is not composable 
since there is a tree component in $G_{2,3}$ but ${\mathcal I}_{2,3}=0$.

\begin{figure}[t]
  \centering
    \begin{tabular}{cc}
    \includegraphics[width=5.9cm,keepaspectratio]{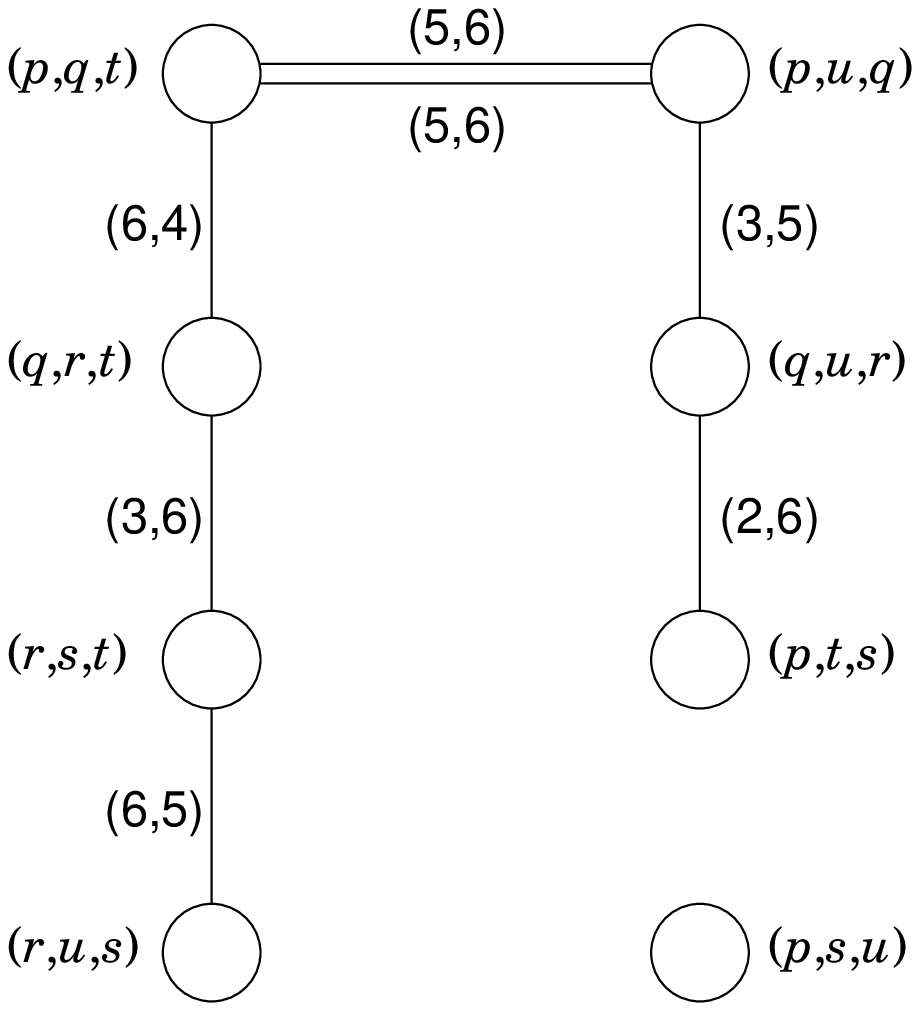} &
    \includegraphics[width=5.9cm,keepaspectratio]{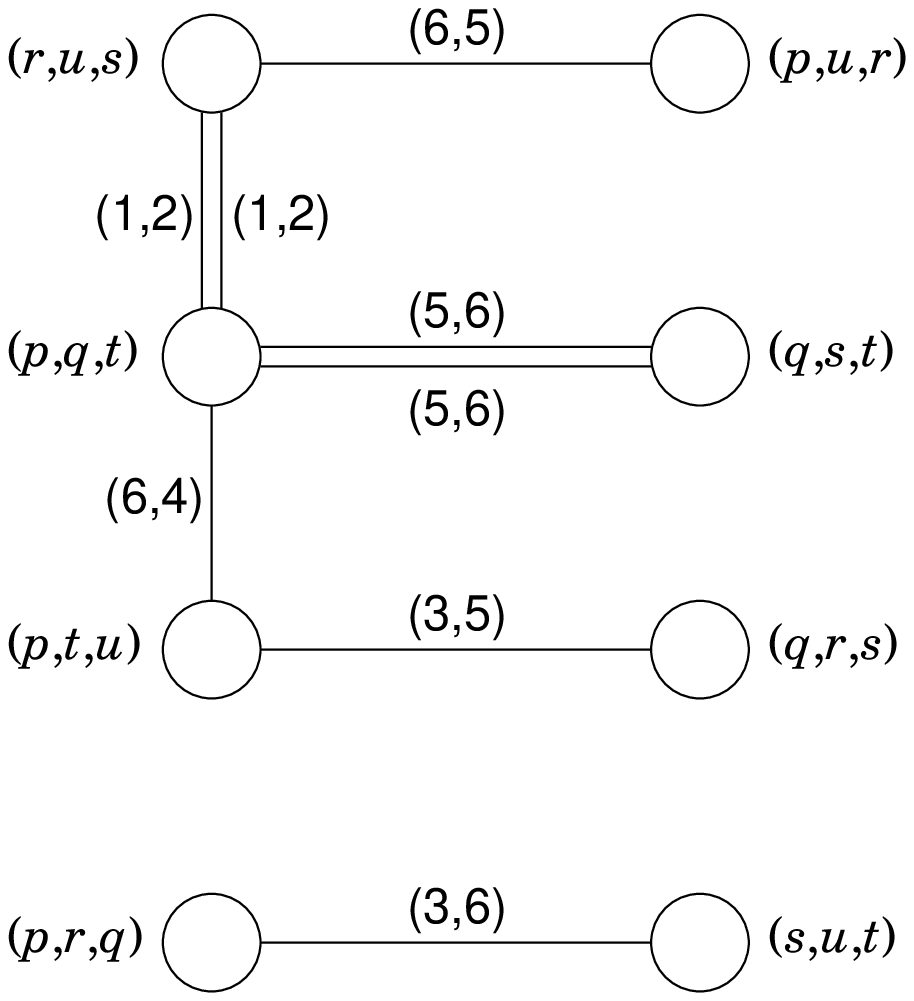} \\
    $G_{1,2}$ & $G_{2,3}$
    \end{tabular}
    \caption{Graphs $G_{1,2}$ and $G_{2,3}$ for the instance $I$ of \eqref{matrix-1}}
    \label{fig:graph-example}
\end{figure}

As a corollary to \propref{Ghave},
we derive a sufficient condition
for an $(i^\ast,j^\ast)$-solid to be composable from a given $I$. 
Let us consider composing an $(i^\ast,j^\ast)$-solid
from only $(i^\ast,j^\ast)$-cubes 
and cubes whose varieties are in a single row $i_0\ne i^\ast$ 
(or a single column $j_0\ne j^\ast$ analogously)
of Conway's table. 
In the row $i_0$,
exactly four varieties are compatible with $(i^\ast,j^\ast)$
from \properref{noshare}. 
Observe a subgraph of $G_{i^\ast,j^\ast}$ 
such that the node set is $T_{i^\ast,j^\ast}$ 
and the edge set is a subset of those coming from cubes of the four varieties. 
Since the four varieties are pairwise incompatible by \properref{noshare},
edges coming from different varieties do not touch with each other. 
Each tree component in the subgraph is an isolated point
or consists of two nodes connected by just one edge. 
The number of tree components of the former kind is 
$2|\{(i_0,j)\in V_{i^\ast,j^\ast}\mid {\mathcal I}_{i_0,j}=0\}|$
and that of the latter kind is 
$|\{(i_0,j)\in V_{i^\ast,j^\ast}\mid {\mathcal I}_{i_0,j}=1\}|$.
Hence, the total number of tree components is: 
\begin{align*}
  &2|\{(i_0,j)\in V_{i^\ast,j^\ast}\mid {\mathcal I}_{i_0,j}=0\}|+|\{(i_0,j)\in V_{i^\ast,j^\ast}\mid {\mathcal I}_{i_0,j}=1\}|\\
  &+0|\{(i_0,j)\in V_{i^\ast,j^\ast}\mid {\mathcal I}_{i_0,j}\ge2\}|\\
  = &\sum_{(i_0,j)\in V_{i^\ast,j^\ast}}(2-\min\{{\mathcal I}_{i_0,j},2\})\\
  = &8-\sum_{(i_0,j)\in V_{i^\ast,j^\ast}}\min\{{\mathcal I}_{i_0,j},2\}.
\end{align*} 
By \propref{Ghave},
if ${\mathcal I}_{i^\ast,j^\ast}$ is no less than the above number, 
then we can compose an $(i^\ast,j^\ast)$-solid. 
More generally,
if ${\mathcal I}_{i^\ast,j^\ast}+\sum_{(i_0,j)\in V_{i^\ast,j^\ast}}\min\{{\mathcal I}_{i_0,j},2\}\ge8$ holds in at least one row $i_0\ne i^\ast$,
or analogously, if ${\mathcal I}_{i^\ast,j^\ast}+\sum_{(i,j_0)\in V_{i^\ast,j^\ast}}\min\{{\mathcal I}_{i,j_0},2\}\ge8$ holds
in at least one column $j_0\ne j^\ast$, 
then we can compose an $(i^\ast,j^\ast)$-solid. 
We define $f(i^\ast,j^\ast)$ as follows:
\begin{align}
f(i^\ast,j^\ast)={\mathcal I}_{i^\ast,j^\ast}+\max\{f_\row(i^\ast,j^\ast),f_\col(i^\ast,j^\ast)\},
\label{eq:f}
\end{align}
where $f_\row$ and $f_\col$ are defined as:
\begin{align}
f_\row(i^\ast,j^\ast)&=\max_{i_0\in[6]\setminus\{i^\ast\}}\{\sum_{(i_0,j)\in V_{i^\ast,j^\ast}}\min\{{\mathcal I}_{i_0,j},2\}\},\label{eq:frow}\\
f_\col(i^\ast,j^\ast)&=\max_{j_0\in[6]\setminus\{j^\ast\}}\{\sum_{(i,j_0)\in V_{i^\ast,j^\ast}}\min\{{\mathcal I}_{i,j_0},2\}\}.\label{eq:fcol} 
\end{align}
\begin{cor}
\label{cor:Ghave}
If $f(i^\ast,j^\ast)\ge8$, then an $(i^\ast,j^\ast)$-solid is composable. 
\end{cor}

\section{Answers to \probexist, \probmax\ and \probmin}
\label{sec:answer}

In this section, we present answers to the three problems,
\probexist, \probmax\ and \probmin.
To draw the answers, we utilize the power of a computer search;
we formulate the problems by CP and IP 
and then solve them using computational packages. 

Let us begin with description of CP and IP. 
The CP technique that we use is
{\em Constraint Satisfaction Problem\/} ({\em CSP\/}). 
In a general setting of CSP,
we are given a set of {\em variables\/},
each of which has its {\em domain\/}, 
and a set of {\em constraints\/}
restricting the values that the variables can simultaneously take.
We are asked to assign a value from its domain to every variable,
in such a way that every constraint is satisfied. 
The Global Constraint Catalog~\cite{CATALOG} summarizes
constraint types that are typically used in CSP. 
The key techniques to solve CSP include {\em backtracking\/},
{\em constraint propagation\/}, and so on~\cite{R.2003}. 

Concerning IP, 
we employ {\em Integer Linear Programming\/} ({\em ILP\/}) among various models;
given a linear function of integer variables called an {\em objective\/}
and a set of linear equalities and/or inequalities representing constraints,
we are asked to assign integers to the variables
so that the objective value is minimized (or maximized)
and that, at the same time, all the constraints are satisfied. 
The key techniques to solve ILP
include {\em linear programming\/}, {\em branch-and-bound\/}, {\em cutting plane\/}, and so on~\cite{W.1998}.

The software choice is significant
since it has a great influence on computation time. 
For CSP, 
we use \SUGAR\ (ver.~2.2.1)~\cite{SUGAR} 
that solves a CSP model by means of 
solving a corresponding SAT model.
\SUGAR\ transforms an input CSP model into
an artificial one by preprocessing,
encodes it into a SAT model
by a certain sophisticated algorithm,
and then runs a SAT solver. 
We use \MINISAT\ (ver.~2.2.1)~\cite{MINISAT} as the SAT solver. 
For ILP, we utilize IBM ILOG CPLEX (ver.~12.6)~\cite{CPLEX}.
Both are recognized as excellent software packages. 

The computation was conducted on a workstation 
that carries Intel$^{\textregistered}$ Core$^{\texttrademark}$ i7-4770 Processor (up to 3.90GHz
by means of Turbo Boost Technology)
and 8GB main memory. The installed OS is Ubuntu 14.04.1.

We formulate \probexist\ and \probmax\ by CSP and \probmin\ by ILP. 
CSP and ILP are so robust that
our three problems can be formulated by either model,
but they have different weak points from each other. 
For example, ILP is not very suitable for formulating 
disjunctive constraints
since we are possibly required to introduce 
a large number of 
artificial variables and constraints. 
Problems \probexist\ and \probmax\ contain
disjunctive constraints,
and in our preliminary study,
ILP models for these two problems are not solved
within $10^4$ seconds;
CSP is more suitable for formulating them. 
On the other hand,
\probmin\ does not contain a disjunctive constraint,
and the ILP model takes only a couple of seconds
to find a minimum universal instance,
while the CSP model takes $4.5\times10^3$ seconds.

\subsection{Problem \probexist}
To say no, it suffices to present a variety subset $V$ of $V_\all$
for which a $V$-generatable instance does not exist. 
Based on our preliminary study, 
we conjecture that $V$ consisting of five varieties in the same row
(or column) should be such a subset.

Let us set $V=\{(1,2),(1,3),(1,4),(1,5),(1,6)\}$. 
We formulate the problem of finding a $V$-generatable instance by CSP,
expecting the software \SUGAR\ to say that such an instance does not exist. 
In the CSP model, 
we have 30 non-negative integer variables,
denoted by $x_{1,2},x_{1,3},\dots,x_{6,5}$,
that represent the distribution of cubes over the 30 varieties,
that is, $x_{i,j}={\mathcal I}_{i,j}$. 
The domain of each $x_{i,j}$ is set as follows:
\begin{align}
  \begin{array}{lcl}
    x_{i,j}\in\{0,1,\dots,8\} &\ \ \ \ &\mathrm{if\ }(i,j)\in V,\\
    x_{i,j}\in\{0,1,2\} && \mathrm{otherwise}.
  \end{array}
  \label{eq:csp-dom}
\end{align}
We assume $x_{i,j}\le 8$ since a ninth cube is redundant. 
Furthermore, we assume $x_{i,j}\le 2$ for any $(i,j)\notin V$
since $(i,j)$-cubes can be used to compose an $(i^\ast,j^\ast)$-solid
$((i^\ast,j^\ast)\in V)$ at most twice and thus a third cube is unnecessary. 
The constraints include:
\begin{description}
\item[(I)] The instance size is at least eight. 
\item[(II)] For each composable variety $(i^\ast,j^\ast)\in V$, 
  $|T|\le|I(T)|$ holds for every subset $T$ of $T_{i^\ast,j^\ast}$. 
\item[(III)] For each non-composable variety $(i^\ast,j^\ast)\notin V$,
  $|T|>|I(T)|$ holds for at least one subset $T$ of $T_{i^\ast,j^\ast}$. 
\end{description}
(II) and (III) come from 
the composability condition of \corref{Bhave}.

CSP deals with a constraint that is represented by a linear inequality of variables. 
(I) and (II) can be expressed by a set of linear inequalities;
For (I),
\begin{align}
  &\sum_{(i,j)\in V_\all} x_{i,j}\ge 8. 
  \label{eq:csp-inst}
\end{align}
For (II),
\begin{align}
  \forall (i^\ast,j^\ast)\in V,\ 
  \forall T\subseteq T_{i^\ast,j^\ast},\ \ \ 
  x_{i^\ast,j^\ast}+\sum_{(i,j)\in V_{i^\ast,j^\ast}:\ T_{i,j}\cap T\ne\emptyset} x_{i,j}\ge |T|, 
  \label{eq:csp-comp}
\end{align}
where the left side of the inequality
is equal to $|I(T)|$.

(III) is a disjunctive constraint. 
To express this, we employ OR constraint:
\begin{align}
  \forall (i^\ast,j^\ast)\notin V,\ \ \ 
\bigvee_{T\subseteq T_{i^\ast,j^\ast}}\big(x_{i^\ast,j^\ast}+\sum_{(i,j)\in V_{i^\ast,j^\ast}:\ T_{i,j}\cap T\ne\emptyset} x_{i,j}< |T|\big). 
\end{align}
Furthermore, to increase the computational efficiency,
we reduce the search space
by introducing the following constraint
that is derived from \corref{Ghave}:
\begin{align*}
  \forall (i^\ast,j^\ast)\notin V,\ \ \ 
  f(i^\ast,j^\ast)\le 7. 
\end{align*}
Observing \eqref{f} to \eqref{fcol},
we express the constraint
by a set of linear inequalities and MIN constraints as follows:
\begin{align}
  &\forall (i^\ast,j^\ast)\notin V,\ \forall i_0\in[6]\setminus\{i^\ast\},
  &x_{i^\ast,j^\ast}+\sum_{(i_0,j)\in V_{i^\ast,j^\ast}}\min\{x_{i_0,j},2\}\le 7,\label{eq:csp-frow}\\
  &\forall (i^\ast,j^\ast)\notin V,\ \forall j_0\in[6]\setminus\{j^\ast\},
  &x_{i^\ast,j^\ast}+\sum_{(i,j_0)\in V_{i^\ast,j^\ast}}\min\{x_{i,j_0},2\}\le 7.\label{eq:csp-fcol}
\end{align}

Finally, our primitive CSP model has 
30 integer variables with domains in \eqref{csp-dom}
and constraints from \eqref{csp-inst} to \eqref{csp-fcol}. 
As a part of preprocessing,
\SUGAR\ transforms this na\"ive model
into an artificial one that has $4.5\times10^4$ integer variables,
$6.3\times10^3$ Boolean variables
and $7.5\times10^5$ constraints,
and then encodes it into a SAT model
that has $5.9\times10^5$ variables and $1.1\times10^7$ clauses. 
After $1.6\times10^3$-second computation,
\MINISAT\ decides that the SAT model is not satisfiable,
which means that no $V$-generatable instance exists. 

\begin{theorem}
\label{thm:exist}
There is a subset $V$ of $V_\all$ such that
a $V$-generatable instance does not exist. 
\end{theorem}

The above argument implies that,
if all five varieties
in a row or a column of Conway's table are composable,
then there is a composable variety outside the row or the column. 
So far, however, we have not gained any insight
into which variety it might be. 

\subsection{Problem \probmax}

It is Berkove et al.~\cite{Berkove.2008}
who first studied \probmax. 
They showed the maximum size of  
infeasible instance to be at least 22,
presenting a concrete example of 
a 22-size infeasible instance as the certificate.
They conjectured the maximum size to be 23,
leaving it open. 

We show by construction
that a 23-size infeasible instance exists. 
We then show that no 24-size infeasible instance exists
by means of CSP.
This answers Berkove et al.'s open problem. 

\begin{theorem} 
  \label{thm:maxinf-lower}
  There is an infeasible instance of size $23$. 
\end{theorem}
\begin{proof}
  We show that the instance $N$ with the following matrix representation is infeasible:
\begin{align}
{\mathcal N}=\left(
\begin{array}{llllll}
0 & 7 & 7 & 7 & 1 & 1\\
0 & 0 & 0 & 0 & 0 & 0 \\
0 & 0 & 0 & 0 & 0 & 0 \\
0 & 0 & 0 & 0 & 0 & 0 \\
0 & 0 & 0 & 0 & 0 & 0 \\
0 & 0 & 0 & 0 & 0 & 0
\end{array}
\right).
\label{eq:maxinf}
\end{align}
The size of $N$ is exactly 23.
We show
that no $(i,j)$
satisfies the composability condition of \propref{Ghave}. 
Recall Properties~\ref{proper:mirror} and \ref{proper:noshare}
of Conway's table. 
For every $(1,j)$ with $j\in\{2,\dots,6\}$,
the graph $G_{1,j}$ consists of
eight isolated points,
while ${\mathcal N}_{1,j}$ is smaller than eight. 
Hence, a $(1,j)$-solid is not composable. 

For $(i,j)=(2,1)$,
the edges appearing in $G_{2,1}$ 
come from the cubes of the varieties $(1,3),\dots,(1,6)$
since $(1,2)$ is the mirror variety of $(2,1)$.
The edges from the same variety
form multi-edges, and
those coming from different varieties
do not touch with each other
since they are pairwise incompatible. 
There are four connected components in $G_{2,1}$,
each of which consists of two nodes. 
Among these, two connected components
contain only one edge respectively,
which are tree components,
but $N$ has no $(2,1)$-cube. 
Hence, a $(2,1)$-solid is not composable.
The remaining cases can be shown in the same way. 
\end{proof}

\begin{lem}
  \label{lem:maxinf}
  There is no infeasible instance of size $24$ or greater. 
\end{lem}
\begin{proof}
Setting $V=\emptyset$,
we show that
a 24-size $V$-generatable instance does not exist by utilizing CSP.  
Here we do not use the inequality constraint \eqref{csp-inst}
that gives a lower bound on the instance size
but an equality constraint that fixes the size to 24:
\begin{align}
\sum_{(i,j)\in V_\all} x_{i,j}=24.
\label{eq:csp-24}
\end{align}
The CSP model in this case
has 30 integer variables with domains in \eqref{csp-dom}
and constraints from \eqref{csp-comp} to \eqref{csp-fcol} and \eqref{csp-24}. 
We run \SUGAR\ to solve the model;
the preprocessed CSP model
has $7.4\times10^4$ integer variables,
$7.7\times10^3$ Boolean variables,
and $1.3\times10^5$ constraints. 
It is then encoded into a SAT model
that has $1.8\times10^6$ variables and $4.2\times10^7$ clauses. 
After $4.4\times10^3$-second computation,
\MINISAT\ reports that no solution exists,
meaning that no 24-size infeasible instance exists. 
Clearly, any larger instance 
is not infeasible. 
\end{proof}

\begin{theorem}
  \label{thm:maxinf}
  The maximum size of infeasible instance is $23$.
  The instance with the matrix representation in $\eqref{maxinf}$
  is an example. 
\end{theorem}

\invis{
In our preliminary study,
we develop an algorithm that enumerates all 24-size instances
to find an infeasible one.  
Of course, the result is the same as above;
the program verifies that no such instance exists. 
}

\subsection{Problem \probmin}

First, we give a lower bound on the size of universal instance. 
We then show that the bound is tight 
by presenting a universal instance of that size. 

\begin{lem}
  \label{lem:minuni-lower}
  A universal instance should contain at least $12$ cubes. 
\end{lem}
\begin{proof}
  Let $U$ be a universal instance. 
  For every variety $(i,j)$ in $V_\all$,
  since an $(i,j)$-solid is composable,
  the number $|E_{i,j}|+{\mathcal U}_{i,j}$
  should be at least eight. 
  The sum of this number over the 30 varieties
  is at least $30\times8=240$. 
  On the other hand,
  by \propref{share}, 
  each cube in $U$
  contributes to the sum in exactly $21$ varieties;
  denoting the variety of the cube by $(i',j')$,
  once for ${\mathcal U}_{i',j'}$,
  and 20 times for an edge in $G_{i,j}$ such that $(i',j')\in V_{i,j}$. 
  Hence we have $|U|\ge 240/21>11$. 
\end{proof}

\begin{theorem}
  \label{thm:minuni}
  The minimum size of universal instance is $12$.
  The instance $U$ 
  with the following matrix representation 
  is an example$:$
\begin{align*}
{\mathcal U}=
\left(
\begin{array}{llllll}
0 & 1 & 1 & 0 & 0 & 0\\
1 & 0 & 1 & 0 & 0 & 0 \\
1 & 1 & 0 & 0 & 0 & 0 \\
0 & 0 & 0 & 0 & 1 & 1 \\
0 & 0 & 0 & 1 & 0 & 1 \\
0 & 0 & 0 & 1 & 1 & 0
\end{array}
\right).
\end{align*}
\end{theorem}
\begin{proof}
The size of $U$ is exactly 12. 
One can verify the universality by 
checking that Proposition~\ref{prop:Bhave} or \ref{prop:Ghave} holds
for all the 30 varieties. 
\end{proof}

We find the minimum universal instance above by
solving the following ILP. 
\begin{align*}
  \begin{array}{lcl}
    \textbf{minimize} &\ \ \ & \sum_{(i,j)\in V_\all} x_{i,j}\\
    \textbf{subject\ to}&& \eqref{csp-comp}\textrm{\ with\ respect\ to\ }V=V_\all\\
    &&\forall(i,j)\in V_\all,\ \ 0\le x_{i,j}\le 8\ \mathrm{and}\ x_{i,j}\in\mathbb{Z}
  \end{array}
\end{align*}
This ILP model has 30 integer variables, $x_{1,2},x_{1,3},\dots,x_{6,5}$,
and $7.7\times10^3$ inequality constraints,
and does not contain the disjunctive constraint (III). 
IBM ILOG CPLEX solves it within a couple of seconds,
providing $\mathcal U$ as output.

\section{Concluding Remarks}
\label{sec:conc}

We solved three problems 
on a generalization of ``Eight Blocks to Madness'' puzzle,
i.e., \probexist, \probmax\ and \probmin,
with the help of a computer search. 
The answers are summarized as Theorems~\ref{thm:exist},
\ref{thm:maxinf} and \ref{thm:minuni}, respectively. 
In particular, \probmax\ 
answers an open problem in \cite{Berkove.2008}.

Smarter proofs are expected,
especially for \probexist\ and \probmax,
since the proofs are computer-assisted;
we admit that ``computer-assisted proof'' is 
a controversial proof technique~\cite{M.2001}.

We leave two open problems. 

\begin{description}
\item[\underline{Open Problem 1}] Characterize subsets $V$ of $V_\all$
  for which a $V$-generatable instance exists.
\end{description}
We have seen that a $V$-generatable instance exists
if $V=\emptyset$ or $V_\all$,
whereas it does not exist if $V$ is a set of five varieties
in a single row (or column) of Conway's table (\thmref{exist}). 
However, we are not sure
what discriminates variety subsets that have generatable instances
and ones that do not have generatable instances. 
\invis{
Clearly a $V$-generatable instance exists when $V=\emptyset$ or $V_\all$. 
On the other hand, it does not exist
if $V$ is a set of 5 varieties in a single row (or column) of Conway's table. 
We found this $V$ by ad-hoc search, 
but are not sure what discriminates $V$'-s that have generatable instances
and ones that do not have generatable instances. 
}

\begin{description}
\item[\underline{Open Problem 2}]
  Characterize subsets $V$ of $V_\all$
  for which every $V$-generatable instance is finite. 
\end{description}
This problem is restricted to $V$
for which a $V$-generatable instance exists. 
Here we do not restrict ${\mathcal I}_{i,j}$ to be no more than eight but permit it to be infinite. 
When $V=\emptyset$, all $V$-generatable instances are finite
because the sizes are no more than 23 (\thmref{maxinf}). 
On the other hand, when $V=V_\all$ or $|V|=1$,
there exists a $V$-generatable instance of infinite size. 
It is interesting to explore what
determines the finiteness of $V$-generatable instance. 

Although being quite simple and classical,
colored cube puzzles still provide us with many mathematical problems. 
We hope 
this paper forms a basis of additional
research in recreational mathematics.

\section*{Acknowledgments}
We gratefully acknowledge very careful and detailed comments given by anonymous reviewers.
A large part of this work was accomplished 
when the author worked for Ishinomaki Senshu University in Japan.
He is deeply grateful to Professor Akira Maruoka
for his proofreading,
enthusiastic guidance on technical writing
and persistent encouragement. 
The author would also like to thank
undergraduate students 
Masaki Ishigaki, Kouta Kakizaki and Kousuke Sakai
for working on preliminary studies. 
This work is partially supported by JSPS KAKENHI Grant Number 25870661.


\bibliographystyle{plain}
\bibliography{bibcube}

\end{document}